\documentclass[envcountsect,envcountsame,runningheads,oribibl]{llncs}

\newcommand\damien[1]{\marginpar{$\dagger$}}

\newcommand{\keywords}[1]{\par\addvspace\baselineskip
\noindent\keywordname\enspace\ignorespaces#1}

\usepackage{lncs-preamble}
\usepackage[utf8]{inputenc}
\usepackage{
enumitem,
mathpartir,
amsmath,amsfonts,
amssymb,%
stmaryrd,cmll,
graphicx,empheq,wrapfig,picinpar,
yhmath,mathabx
} %
\usepackage[mathcal]{euscript}
\usepackage{fullshort}
\setboolean{fullpaper}{true}

\newlist{myaxioms}{enumerate}{10}
\setlist[myaxioms,1]{label=(P\arabic*)}
\setlist[myaxioms]{resume}

\newlist{enumeratei}{enumerate}{10}
\setlist[enumeratei]{label=\emph{(\roman*)}}

\DeclareMathOperator{\ob}{ob}

\DeclareMathOperator{\dom}{dom}
\DeclareMathOperator{\cod}{cod}
\DeclareMathOperator{\domv}{dom_{\mathit{v}}}
\DeclareMathOperator{\codv}{cod_{\mathit{v}}}
\DeclareMathOperator{\domh}{dom_{\mathit{h}}}
\DeclareMathOperator{\codh}{cod_{\mathit{h}}}

\newcommand{\restr}[2]{#1_{|#2}}

\newcommand{\bigcat}[1]{\ensuremath{\mathsf{#1}}}
\newcommand{\maji}[1]{\ensuremath{\mathbb{#1}}}

\newcommand{\commentthis}[1]{}

\newcommand{\Psh}[1]{\widehat{#1}}

\newcommand{\Cospan}[1]{\bigcat{Cospan} (#1)}

\newcommand{\into}{\hookrightarrow}

\newcommand{\ot}{\leftarrow}
\newcommand{\xto}[1]{\xrightarrow{#1}}
\newcommand{\xot}[1]{\xleftarrow{#1}}

\newcommand{\xinto}[1]{\xhookrightarrow{#1}}

\renewcommand{\hat}[1]{\widehat{#1}}

\def\framed{%
\setbox0=\vbox\bgroup%
\advance\hsize by -2\fboxsep\advance\hsize by -2\fboxrule%
\linewidth=\hsize%
}
\def\endframed{%
\egroup\noindent\framebox[\textwidth]{\box0}\vspace*{1mm}}

\usepackage{tikz,tikz-cats}
\tikzset{todim/.style = {decoration={markings, mark=at position .5 with %
      {\draw (-1pt,-1pt) rectangle (1pt,1pt);}},postaction={decorate}}}
\renewcommand{\diaggrandhauteur}{.6}
\renewcommand{\diaggrandlargeur}{1.3}

\usepackage{color}
\definecolor{dkgreen}{rgb}{0,0.2,0}

\newcommand{\A}{\maji{A}}

\newcommand{\B}{\maji{B}}
\newcommand{\Beh}[1]{\Behfun_{#1}}
\newcommand{\Behfun}{\bigcat{B}}

\newcommand{\C}{\maji{C}}

\newcommand{\tick}{\daimon}

\newcommand{\Chat}{\hat{\maji{C}}}
\newcommand{\Chatf}{\Pshf{\maji{C}}}

\newcommand{\D}{\maji{D}}

\renewcommand{\DH}{\D_{H}}

\newcommand{\Dh}{\D_h}

\newcommand{\Dv}{\D_v}

\newcommand{\ccs}{\mathit{CCS}}
\newcommand{\Dccs}{\maji{D}^{\scriptscriptstyle \ccs}}

\newcommand{\Dccsh}{\Dccs_h}
\newcommand{\Dccsv}{\Dccs_v}

\newcommand{\DI}{\I}

\newcommand{\E}{\maji{E}}

\newcommand{\EVi}{\E^{\V}}

\newcommand{\F}{\maji{F}}

\newcommand{\G}{\maji{G}}

\newcommand{\I}{\maji{I}}

\newcommand{\LLL}{\mathcal{L}}

\newcommand{\griso}[1]{[#1]}
\newcommand{\MMMB}{\griso{\B}}

\newcommand{\V}{\maji{V}}

\newcommand{\yoneda}{\bigcat{y}}

\renewcommand{\SS}{\bigcat{S}}
\newcommand{\SSS}{\mathcal{S}}

\newcommand{\SSn}{\SS_{[n]}}

\newcommand{\SSX}{\SS_X}
\newcommand{\SSY}{\SS_Y}

\newcommand{\TT}{\bigcat{T}} %
\renewcommand{\tt}{\bigcat{t}} %
\newcommand{\TTT}{\mathcal{T}} %

\newcommand{\Set}{\bigcat{Set}}
\newcommand{\set}{\bigcat{set}}
\newcommand{\ford}{\bigcat{ford}}

\newcommand{\OPsh}[1]{\wideparen{#1}}
\newcommand{\Pshf}[1]{\widehat{#1}^{{}_f}}

\newcommand{\Gph}{\bigcat{Gph}}

\newcommand{\Nat}{\mathbb{N}} 

\newcommand{\un}{I}
\newcommand{\para}{\mathbin{\mid}}

\newcommand{\translfun}{\llbracket - \rrbracket}
\newcommand{\Transl}[1]{\llparenthesis #1 \rrparenthesis}
\newcommand{\Translfun}{\Transl{-}}
\newcommand{\rond}{\circ}
\newcommand{\vrond}{\mathbin{\scriptstyle \bullet}}

\newcommand{\id}{\mathit{id}}
\newcommand{\iso}{\cong}

\newcommand{\impll}{\multimap}
\newcommand{\tens}{\otimes}

\newcommand{\ens}[1]{\{ #1 \}}

\newcommand{\aalt}{\mathrel{|}}

\newcommand{\op}[1]{#1^{\mathit{op}}}

\newcommand{\Gam}{\Gamma}

\newcommand{\wbisim}{\mathrel{\simeq}}

\newcommand{\wbisima}{\mathrel{\simeq_A}}

\newcommand{\faireq}{\mathrel{\sim_f}}

\newcommand{\faireqs}{\mathrel{\sim_{f,s}}}

\newcommand{\bureaucratic}[1]{}

\newcommand{\abar}{\overline{a}}

\newcommand{\sender}{\epsilon}
\newcommand{\receiver}{\rho}

\newcommand{\iotapos}[1]{o_{#1}}
\newcommand{\iotaneg}[1]{\iota_{#1}}
\newcommand{\iotaposni}{\iotapos{n,i}}
\newcommand{\iotaposmj}{\iotapos{m,j}}

\newcommand{\iotanegni}{\iotaneg{n,i}}
\newcommand{\iotanegmj}{\iotaneg{m,j}}

\newcommand{\paraof}[1]{\pi_{#1}}

\newcommand{\paralof}[1]{\pi^l_{#1}}
\newcommand{\pararof}[1]{\pi^r_{#1}}
\newcommand{\paran}{\paraof{n}}
\newcommand{\paraln}{\paralof{n}}
\newcommand{\pararn}{\pararof{n}}
\newcommand{\paragam}{\paraof{\Gam}}
\newcommand{\paralgam}{\paralof{\Gam}}
\newcommand{\parargam}{\pararof{\Gam}}

\newcommand{\nuof}[1]{\nu_{#1}}
\newcommand{\nun}{\nuof{n}}
\newcommand{\nugam}{\nuof{\Gam}}
\newcommand{\tickof}[1]{\tick_{#1}}
\newcommand{\tickn}{\tickof{n}}
\newcommand{\tickgam}{\tickof{\Gam}}
\newcommand{\tauof}[4]{\tau_{#1,#2,#3,#4}}
\newcommand{\taunimj}{\tauof{n}{i}{m}{j}}
\newcommand{\taumjni}{\tauof{m}{j}{n}{i}}

\newcommand{\Pl}{\mathrm{Pl}}

\newcommand{\lts}{\textsc{lts}}
\newcommand{\anlts}{an \lts{}}
\newcommand{\ltss}{\lts{s}}
\newcommand{\Lts}{\textsc{Lts}}
\newcommand{\Ltss}{\Lts{s}}

\newcommand\independent{\protect\mathpalette{\protect\independenT}{\perp}} 
\def\independenT#1#2{\mathrel{\rlap{$#1#2$}\mkern2mu{#1#2}}} 
\newcommand{\bbot}{\mathord{\independent}}

\newcommand{\daimon}{\heartsuit}

\newcommand{\state}{\sigma}

\newcommand{\new}{\nu}

\renewcommand{\with}[1]{\langle #1 \rangle}

\newcommand{\exta}[2]{\overline{#2}}
\newcommand{\extafun}[1]{\overline{(-)}}

\newcommand{\pbang}[2]{#1_{!}(#2)}

\newcommand{\pbkchi}[1]{#1^\LLL}

\newcommand{\SSSL}{\pbkchi{\SSS}}

\newcommand{\rc}[1]{#1^\A}
\newcommand{\SSSA}{\rc{\SSS}}
\newcommand{\TTTA}{\rc{\TTT}}
\newcommand{\SSSAccs}{\rc{\SSS_{\Dccs}}}
\newcommand{\TTTAccs}{\rc{\TTT_{\Dccs}}}

\tikzset{history/.style = {-open triangle 45}}

\newcommand{\vdashdefinite}{\vdash_{\mathsf{D}}}

\renewcommand{\EVi}{\V}

\begin{document}

\mainmatter  

\title{Full abstraction for fair testing in CCS}


\author{Tom Hirschowitz%
\thanks{Partially
    funded by the French ANR projet blanc ``Formal Verification of
    Distributed Components'' PiCoq ANR 2010 BLAN 0305 01 and CNRS PEPS
    CoGIP.}%
}
\authorrunning{T.\ Hirschowitz}

\institute{CNRS and Universit\'e de Savoie}

\maketitle

\begin{abstract}
In previous work with Pous, we defined a semantics for
  CCS which may both be viewed as an innocent presheaf semantics and
  as a concurrent game semantics.  It is here proved that a
  behavioural equivalence induced by this semantics on CCS processes
  is fully abstract for fair testing equivalence.

  The proof relies on a new algebraic notion called \emph{playground},
  which represents the `rule of the game'. From any playground, two
  languages, equipped with labelled transition systems, are derived,
  as well as a strong, functional bisimulation between them.

\end{abstract}

\section{Introduction}
\subsubsection*{Motivation and previous work}
Innocent game semantics, invented by Hyland and
Ong~\cite{DBLP:journals/iandc/HylandO00}, led to fully abstract models
for a variety of functional languages, where programs are interpreted
as strategies in a game.  Presheaf
models~\cite{DBLP:conf/lics/JoyalNW93,DBLP:conf/csl/CattaniW96} were
introduced by Joyal et al. as a semantics for process algebras, in
particular Milner's CCS~\cite{Milner80}.  Previous work with
Pous~\cite{2011arXiv1109.4356H} (HP) proposes a semantics for CCS,
which reconciles these apparently very different approaches.  Briefly,
(1) on the one hand, we generalise innocent game semantics to both
take seriously the possibility of games with more than two players and
consider strategies which may accept plays in more than one way; (2)
on the other hand, we refine presheaf models to take parallel
composition more seriously. This leads to a model of CCS which may
both be seen as a concurrent game semantics, and as an innocent
presheaf model, as we now briefly recall.

To see that presheaf models are a concurrent, non-innocent
variant of game semantics, recall that the base category, say $\C$,
for such a presheaf model typically has as objects sequences of
labels, or configurations in event structures, morphisms being given
by prefix inclusion. Such objects may be understood as plays in some
game. Now, in standard game semantics, a strategy is a prefix-closed
(non-empty) set of plays. Unfolding the definition, this is the same
as a functor $\op\C \to 2$, where $2$ is the poset category $0 \leq
1$: the functor maps a play to $1$ when it is accepted by the
strategy, and to $0$ otherwise.  It is known since Harmer and
McCusker~\cite{DBLP:conf/lics/HarmerM99} that this notion of
strategy does not easily adapt to non-determinism or concurrency.
Presheaf semantics only slightly generalises it by allowing strategies
to accept a play in several ways.  A strategy $S$ now maps each play
$p$ to a \emph{set} $S (p)$. The play is accepted when $S (p)$ is
non-empty, and, because there are then no functions $S (p) \to
\emptyset$, being accepted remains a prefix-closed property of plays.
The passage from $2$ to more general sets allows to express
branching-time semantics.

This links presheaf models with game models, but would be of little
interest without the issue of \emph{innocence}. Game models, indeed,
do not always accept \emph{any} prefix-closed set of plays $S$ as a
strategy: they demand that any choice of move in $S$ depends only on
its \emph{view}. E.g., consider the CCS process $P = (a | (b \oplus
c))$, where $\oplus$ denotes internal choice, and a candidate strategy
accepting the plays $\epsilon, (a), (b),(c), (ab),$ but not
$(ac)$. This strategy refuses to choose $c$ after $a$ has been
played. Informally, there are two players here, one playing $a$ and
the other playing $b \oplus c$; the latter should have no means to
know whether $a$ has been played or not.  We want to rule out this
strategy on the grounds that it is not innocent.  

Our technical solution for doing so is to refine the notion of play,
making the number of involved players more explicit. Plays still form
a category, but they admit a subcategory of \emph{views}, which
represent a single player's possible perceptions of the game. This
leads us to two equivalent categories of strategies.  In the first,
strategies are presheaves on views.  In the second category,
strategies are certain presheaves on arbitrary plays, satisfying an
innocence condition.  Parallel composition, in the game semantical
sense, is best understood in the former category: it merely amounts to
copairing. Parallel composition, in the CCS sense, which in standard
presheaf models is a complex operation based on some labelling of
transitions or events, is here just a move in the game.  The full
category of plays is necessary for understanding the global behaviour
of strategies. It is in particular needed to define our semantic
variant of fair testing equivalence, described below. One may think of
presheaves on views as a syntax, and of innocent presheaves on plays
as a semantics.  The combinatorics of passing from local (views) to
global (arbitrary plays) are dealt with by right Kan extension.

\subsubsection*{Discussion of main results}
In this paper, we further study the semantics of HP, to demonstrate
how close it is to operational semantics. For this, we provide two
results. The most important, in the author's view, is full abstraction
w.r.t.\ \emph{fair testing semantics}. But the second result might be
considered more convincing by many: it establishes that our
semantics is fully abstract w.r.t.\ weak bisimilarity.
The reason why it is here considered less important is that it relies
on something external to the model itself, namely \anlts{} for
strategies, constructed in an \emph{ad hoc} way. 
Considering that a process calculus is defined by its reduction
semantics, rather than by its possibly numerous \ltss{}, testing
equivalences, which rely on the former, are more intrinsic than
various forms of bisimilarity.
 
 

Now, why consider fair testing among the many testing equivalences?
First of all, let us mention that we could probably generalise our
result to any reasonable testing equivalence.  Any testing equivalence
relies on a `testing predicate' $\bot$.  E.g., for fair testing, it is
the set of processes from which any unsuccessful, finite reduction
sequence extends to a successful one.  We conjecture that for any
other predicate $\bot'$, if $\bot'$ is stable under weak bisimilarity,
i.e, $P \wbisim Q \in \bot'$ implies $P \in \bot'$, then we may
interpret the resulting equivalence in terms of strategies, and get a
fully abstract semantics.  However, this paper is already quite
complicated, and pushes generalisation rather far in other respects
(see below).  We thus chose to remain concrete about the considered
equivalence.  It was then natural to consider fair testing, as it is
both one of the most prominent testing equivalences, and one of the
finest.  It was introduced independently by Natarajan and
Cleaveland~\cite{DBLP:conf/icalp/NatarajanC95}, and by Brinksma et
al.~\cite{DBLP:conf/concur/BrinksmaRV95,DBLP:journals/iandc/RensinkV07}
(under the name of \emph{should} testing in the latter paper), with
the aim of reconciling the good properties of observation
congruence~\cite{Milner89} w.r.t.\ divergence, and the good properties
of previous testing equivalences~\cite{DBLP:journals/tcs/NicolaH84}
w.r.t.\ choice. Typically, $a.b + a.c$ and $a.(b \oplus c)$ (where $+$
denotes guarded choice and $\oplus$ denotes internal choice) are not
observation congruent, which is perceived as excessive discriminating
power of observation congruence. Conversely, $(!\tau) \para a$ and $a$
are not must testing equivalent, which is perceived as excessive
discriminating power of must testing equivalence. Fair testing
rectifies both defects, and has been the subject of further
investigation, as summarised, e.g., in Cacciagrano et
al.~\cite{DBLP:journals/corr/abs-0904-2340}.

\subsubsection*{Overview}
We now give a bit more detail on the contents, warning the reader that
this paper is only an extended abstract, and that more technical
details may be found in a (submitted) long
version~\cite{HirschoDoubleCats}.  After recalling the game from HP in
Section~\ref{sec:game}, as well as strategies and our semantic fair
testing equivalence $\faireq$ in Section~\ref{sec:strats}, we prove
that the translation $\Translfun$ of HP from CCS to strategies is such
that $P \faireqs Q$ iff $\Transl{P} \faireq \Transl{Q}$, where
$\faireqs$ is standard fair testing equivalence
(Theorem~\ref{thm:final}).

Our first attempts at proving this where obscured by easy, yet lengthy
case analyses over moves. This prompted the search for a way of
factoring out what holds `for all moves'. The result is the notion of
\emph{playground}, surveyed in Section~\ref{subsec:playgrounds}. It is
probably not yet in a mature state, and hopefully the axioms will
simplify in the future.  We show how the game recalled above organises
into such a playground $\Dccs$.  We then develop the theory in
Section~\ref{subsec:ltss}, defining, for any playground $\D$, two
\ltss{}, $\TTT_\D$ and $\SSS_\D$, of \emph{process terms} and
\emph{strategies}, respectively, over an alphabet $\F_\D$.  We then
define a map $\translfun \colon \TTT_\D \to \SSS_\D$ between them,
which we prove is a strong bisimulation.

Returning to the case of CCS in Section~\ref{subsec:cob}, we obtain
that $\SSS_{\Dccs}$ indeed has strategies as states, and that
$\faireq$ may be characterised in terms of this \lts{}. Furthermore,
unfolding the definition of $\TTT_{\Dccs}$, we find that its states
are terms in a language containing CCS. So, we have maps $\ob(\ccs)
\xinto{\theta} \ob (\TTT_{\Dccs}) \xto{\translfun} \ob
(\SSS_{\Dccs}),$ where $\ob$ takes the set of vertices, and with
$\translfun \rond \theta = \Translfun$.  Now, a problem is that $\ccs$
and the other two are \ltss{} on different alphabets, respectively
$\A$ and $\F_{\Dccs}$. We thus define morphisms $\A \xot{\xi} \LLL
\xto{\chi} \F_{\Dccs}$ and obtain by successive change of base
(pullback when rewinding an arrow, postcomposition when following one)
a strong bisimulation $\translfun \colon \TTTA_{\Dccs} \to
\SSSA_{\Dccs}$ over $\A$.  We then prove that $\theta$, viewed as a
map $\ob (\ccs) \into \ob (\TTTA_{\Dccs})$, is included in weak
bisimilarity, which yields for all $P$, $P \wbisim_\A \Transl{P}$
(Corollary~\ref{cor:wbisim}).  Finally, drawing inspiration from
Rensink et al.~\cite{DBLP:journals/iandc/RensinkV07}, we prove that
$\ccs$ and $\SSSA_{\Dccs}$ both have enough $\A$-\emph{trees}, in a
suitable sense, and that this, together with
Corollary~\ref{cor:wbisim}, entails the main result.

\subsubsection*{Related work}
Trying to reconcile two mainstream approaches to denotational
semantics, we have designed a (first version of a) general framework
aiming at an effective theory of programming languages. Other such
frameworks
exist~\cite{DBLP:conf/lics/Nipkow91,PlotkinSOS,plotkin:turi:bialgebraic,DBLP:conf/birthday/GadducciM00,Bruni:ccdc,DBLP:conf/fossacs/BonsangueRS09,HIRSCHOWITZ:2010:HAL-00540205:2,DBLP:conf/wollic/Ahrens12},
but most of them, with the notable exception of Kleene coalgebra,
attempt to organise the traditional techniques of syntax with variable
binding and reduction rules into some algebraic structure.  Here, as
in Kleene coalgebra, syntax and its associated \lts{} are derived
notions.  Our approach may thus be seen as an extension of Kleene
coalgebra to an innocent/multi-player setting, yet ignoring
quantitative aspects.

In another sense of the word `framework', recent work of Winskel and
colleagues~\cite{RideauW} investigates a general notion of concurrent
game, based on earlier work by Melli\`es~\cite{Mellies04}.  In our
approach, the idea is that each programming language is interpreted as
a playground, and that morphisms of playgrounds denote translations
between languages.  Winskel et al., instead, construct a (large) bicategory, into which
each programming language should embed. Beyond this crucial
difference, both approaches use presheaves and factorisation systems,
and contain a notion of innocent, concurrent strategy. The precise
links between the original notion of innocence, theirs, and ours
remain to be better investigated.

Melli\`es's work~\cite{DBLP:conf/lics/Mellies12}, although in a
deterministic and linear setting, incorporates some `concurrency' into
plays by presenting them as string diagrams. Our innocentisation
procedure further bears some similarity with Harmer et
al.'s~\cite{DBLP:conf/lics/HarmerHM07} presentation of innocence based
on a distributive law.  Hildebrandt's approach to fair testing
equivalence~\cite{DBLP:journals/tcs/Hildebrandt03} uses closely
related techniques, e.g., presheaves and sheaves --- indeed, our
innocence condition may be viewed as a sheaf condition.  However, (1)
his model falls in the aforementioned category of presheaf models for
which parallel composition is a complex operation; and (2) he uses
sheaves to correctly incorporate infinite behaviour in the model,
which is different from our notion of innocence.  Finally, direct
inspiration is drawn from Girard~\cite{DBLP:journals/mscs/Girard01},
one of whose aims is to bridge the gap between syntax and semantics.

\subsubsection*{Perspectives}
We plan to adapt our semantics to more complicated calculi like $\pi$,
the Join and Ambients calculi, functional calculi, possibly with extra
features (e.g., references, data abstraction, encryption), with a view
to eventually generalising it. Preliminary investigations already led
to a playground for $\pi$, whose adequacy remains to be
established. More speculative directions include
(1) defining a notion of morphisms for playgrounds, which should
  induce translations between strategies, and find sufficient
  conditions for such morphisms to preserve, resp.\ reflect testing
  equivalences;
(2) generalising playgrounds to apply them beyond programming
  language semantics; in particular, preliminary work shows that
  playgrounds easily account for cellular automata; this raises the
  question of how morphisms of playgrounds would compare with existing
  notions of simulations between cellular
  automata~\cite{DBLP:journals/tcs/DelormeMOT11};
(3) trying and recast the issue of deriving transition systems
  (\ltss{}) from
  reductions~\cite{DBLP:conf/concur/Sewell98}
  in terms of playgrounds.

\subsubsection*{Notation}
$\Set$ is the category of sets; $\set$ is a skeleton of the category
of finite sets, namely the category of finite ordinals and arbitrary
maps between them; $\ford$ is the category of finite ordinals and
monotone maps between them.  For any category $\C$, $\Psh{\C} =
[\op\C,\Set]$ denotes the category of presheaves on $\C$, while
$\Pshf{\C} = [\op\C,\set]$ and $\OPsh{\C} = [\op\C,\ford]$
respectively denote the categories of presheaves of finite sets and of
finite ordinals.  One should distinguish, e.g., `presheaf of finite
sets' $\op\C \to \set$ from `finite presheaf of sets' $F \colon \op\C
\to \Set$. The latter means that the disjoint union $\sum_{c \in \ob
  (\C)} F (c)$ is finite.  Throughout the paper, any finite ordinal
$n$ is seen as $\ens{1, \ldots, n}$ (rather than $\ens{0, \ldots,
  n-1}$).

The notion of \lts{} that we'll use here is a little more general than
the usual one, but this does not change much. We thus refer to the
long version for details. Let us just mention that we work in the
category $\Gph$ of reflexive graphs, and that the category of \ltss{}
over $A$ is for us the slice category $\Gph / A$. \Ltss{} admit a
standard change of base functor given by pullback, and its left
adjoint given by postcomposition. Given any \lts{} $p \colon G \to A$,
an edge in $G$ is \emph{silent} when it is mapped by $p$ to an
identity edge.  This straightforwardly yields a notion of weak
bisimilarity over $A$, which is denoted by $\wbisim_A$.

Our (infinite) CCS terms are coinductively generated by the typed grammar
\begin{mathpar}
  \inferrule{\Gam \vdash P \\ \Gam \vdash Q}{\Gam \vdash P|Q} \and %
  \inferrule{\Gam,a \vdash P}{\Gam \vdash \nu a. P} \and %
  \inferrule{\ldots \\ \Gam \vdash P_i \\ \ldots}{\Gam \vdash {\textstyle \sum}_{i \in n} \alpha_i.P_i}~(n \in \Nat) \,,%
\end{mathpar}%
where $\alpha_i$ is either $a$, $\abar$, for $a \in \Gam$, or $\tick$.
The latter is a `tick' move used in the definition of fair testing
equivalence.  As a syntactic facility, we here understand $\Gam$ as
ranging over $\Nat$, i.e., the free names of a process always are $1
\ldots n$ for some $n$. E.g., $\Gam,a$ denotes just $n+1$, and
$a \in \Gam$ means $a \in \ens{1, \ldots, \Gam}$.

\begin{definition}\label{defn:A}
  Let $\A$ be the reflexive graph with vertices given by finite
  ordinals, edges $\Gam \to \Gam'$ given by $\emptyset$ if $\Gam \neq
  \Gam'$, and by $\Gam + \Gam + \ens{\id,\tick}$ otherwise,  $\id
  \colon \Gam \to \Gam$ being the identity edge on $\Gam$.
  Elements of the first summand are denoted by $a
  \in \Gam$, while elements of the second summand are denoted by
  $\abar$.
\end{definition}
We view terms as a graph $\ccs{}$ over $\A$ with the usual transition
rules.  The graph $\A$ only has `endo'-edges; some \ltss{} below do
use more general graphs.

\section{Recalling the game}\label{sec:game}

\subsection{Positions, Moves, and Plays}\label{subsec:posmovesplays}
In this section, we define plays in our game.  For lack of space, we
cannot be completely formal. A formal definition, with a gentle
introduction to the required techniques, may be found in HP (Section
3).  Here is a condensed account. We start by defining a category
$\C$.  Then, \emph{positions} in our game are defined to be particular
finite presheaves in $\Chatf$. \emph{Moves} in our game are defined as
certain \emph{cospans} $X \xto{s} M \xot{t} Y$ in $\Chatf$, where $t$
indicates that $Y$ is the \emph{initial} position of the move, while
$s$ indicates that $X$ is the \emph{final} position. \emph{Plays} are
then defined as finite composites of moves in the bicategory
$\Cospan{\Chatf}$ of cospans in $\Chatf$. By construction, positions
and plays form a subbicategory, called $\Dccsv$.

In order to motivate the definition of our base category $\C$, recall
that (directed, multi) graphs may be seen as presheaves over the
category freely generated by the graph with two objects $\star$ and
$[1]$, and two edges $s,t \colon \star \to [1]$. Any presheaf $G$
represents the graph with vertices in $G(\star)$ and edges in $G[1]$,
the source and target of any $e \in G[1]$ being respectively $G(s)(e)$
and $G(t)(e)$. A way to visualise how such presheaves represent graphs
is to compute their \emph{categories of elements}~\cite{MM}. Recall
that the category of elements $\int G$ for a presheaf $G$ over $\C$
has as objects pairs $(c,x)$ with $c \in \C$ and $x \in F(c)$, and as
morphisms $(c,x) \to (d,y)$ all morphisms $f \colon c \to d$ in $\C$
such that $F(f)(y) = x$. This category admits a canonical functor
$\pi_F$ to $\C$, and $F$ is the colimit of the composite $\int F
\xto{\pi_F} \C \xto{\yoneda} \Chat$ with the Yoneda embedding. Hence,
e.g., the category of elements for the representable presheaf over
$[1]$ is the poset $(\star, s) \to ([1],\id_{[1]}) \ot (\star, t)$,
which could be pictured as
\diagramme[stringdiag={0.1}{0.6}]{baseline=(A.south)}{%
  \path[-,draw] %
  (A) edge (E) %
  (B) edge (E) %
  ; %
  \node at ($(B.south east) + (.1,0)$) {,} ;%
}{%
  \joueur{A} \& \node[regular polygon,anchor=center,regular polygon
  sides=3,fill,minimum size=3pt,draw,rotate=-90] (E) {}; \&
  \joueur{B} %
}{%
} \hspace*{-.7em}
thus recovering some graphical intuition.

We now define our base category $\C$. Let us first give the raw
definition, and then explain. $\C$ is freely generated from the graph
$\G$, defined as follows, plus some equations. As objects, $\G$ has %
(1) an object $\star$, %
(2) an object $[n]$ for all $n \in \Nat$, %
(3) objects $\iotaposni$ (output), $\iotanegni$ (input), $\nun$
(channel creation), $\paraln$ (left fork), $\pararn$ (right fork), $\paran$ (fork), $\tickn$ (tick),
$\taunimj$ (synchronisation), %
for all $i \in n, j \in m, n,m \in \Nat$. %
$\G$ has edges, for all $n$, %
(1) $s^{n}_1, \ldots, s^{n}_n \colon \star \to [n]$, %
(2) $s^{c},t^{c} \colon [n] \to c$, for all $c \in
\ens{\paraln,\pararn,\tickn} \cup (\cup_{i \in n} \ens{\iotaposni,
  \iotanegni})$, %
(3) $[n+1] \xto{s^{\nun}} \nun \xot{t^{\nun}} [n]$, %
(4) $\paraln \xto{l^n} \paran \xot{r^n} \pararn$, %
$\iotaposni \xto{\sender^{n,i,m,j}} \taunimj \xot{\receiver^{n,i,m,j}}
\iotanegmj$, for all $i \in n, j \in m$. In the following, we omit
superscripts when clear from context.  As equations, we require, for
all $n$, $m$, $i \in n$, and $j \in m$, (1) $s^c \rond s^n_i = t^c
\rond s^n_i$, %
(2) $s^{\nun} \rond s^{n+1}_i = t^{\nun} \rond s^n_i$, %
(3) $l \rond t = r \rond t$, %
(4) $\sender \rond t \rond s_i = \receiver \rond t \rond s_j$. %

\begin{wrapfigure}{r}{0pt}
      \diagramme[stringdiag={.8}{1.3}]{}{%
    \path[-,draw] %
    (a) edge (j1) %
    (c) edge (j1) %
    (b) edge (j1) %
    ; %
}{%
  \node (s_1) {$(\star, s_1)$}; \& \node (s_2) {$(\star, s_2)$}; \& \node (s_3) {$(\star, s_3)$}; \\
    \& \node (id) {$([3], \id_{[3]})$}; \\
    \canal{a}     \& \canal{b} \&  \canal{c} \\
    \& \joueur{j1}
    }{%
      (s_1) edge (id) %
      (s_2) edge (id) %
      (s_3) edge (id) %
    }
\end{wrapfigure}
In order to explain this seemingly arbitrary definition, let us
compute a few categories of elements for representable presheaves. Let
us start with an easy one, that of $[3]$ (we implicitly identify any
$c \in \C$ with $\yoneda c$). An easy computation shows that it is the
poset pictured above. We will think of it as a position with one
player $([3],\id_{[3]})$ connected to three channels, and draw it as
above, where the bullet represents the player, and circles represent
channels. (The graphical representation is slightly ambiguous, but
nevermind.) In particular, elements over $[3]$ represent ternary
players, while elements over $\star$ represent channels.
\emph{Positions} are finite presheaves empty except perhaps on $\star$
and $[n]$'s. Let $\Dccsh$ be the subcategory of $\Chatf$ consisting of
positions and monic arrows between them.

A more difficult category of elements is that of $\paraof{2}$. It is the poset generated by
the graph on the left:
  \begin{center}
    \diag (.4,.3) {%
      |(lt1)| l s s_1 = r s s_1 \& \& |(lt)| l s \& \& |(rt)| r s \& \& |(lt2)| l s s_2 = r s s_2 \\
      \& \& |(l)| l \& |(para)| \id_{\paraof{2}} \& |(r)| r \\ 
      |(ls1)| l t s_1 = r t s_1 \& \& \& |(ls)| l t = r t \&  \& \& |(ls2)| l t s_2 = r t s_2 %
    }{%
      (lt1) edge[identity] (ls1) %
      edge (lt) %
      edge[bend left=15] (rt) %
      (lt2) edge[identity] (ls2) %
      edge (rt) %
      edge[bend right=15,fore] (lt) %
      (ls1) edge (ls) %
      (ls2) edge (ls) %
      (ls) edge (l) edge (r) %
      (lt) edge (l) %
      (rt) edge (r) %
      (l) edge (para) %
      (r) edge (para) %
    }
    \hfil
          \diagramme[stringdiag={.18}{.33}]{}{
    \node[diagnode,at= (s1.south east)] {\ \ \ .} ; %
  }{%
    \& \& \&  \& \joueur{t_2} \\
    \canal{t0} \& \& \& \& \& \& \& \canal{t1} \& \& \\ 
    \& \& \joueur{t_1}  \&  \\
    \& \ \& \\
    \coord{i0} \& \& \& \couppara{para} \& \& \& \& \coord{i1} \& \& \\ 
    \& \ \& \\
    \& \&  \\
    \canal{s0} \& \& \& \joueur{s} \& \& \& \& \canal{s1} \& \& 
  }{%
    (para) edge[-] (t_1)
    (para) edge[-] (t_2) %
    (t0) edge[-] (t_1) %
    (t0) edge[-] (t_2) %
    (t1) edge[-,fore] (t_1) %
    (t1) edge[-] (t_2) %
    (s0) edge[-] (s) %
    (s1) edge[-] (s) %
    (s0) edge[-] (t0) %
    (s1) edge[-] (t1) %
    (s) edge[-] (para) %
  }  
  \end{center}
  We think of it as a binary player ($l t$) forking into two players
  ($l s$ and $r s$), and draw it as on the right. The vertical edges
  on the outside are actually identities: the reason we draw separate
  vertices is to identify the top and bottom parts of the picture as
  the respective images of both legs of the following cospan.  First,
  consider the inclusion $[2] \para [2] \into \paraof{2}$: its domain
  is any pushout of $[s_1,s_2] \colon (\star + \star) \to [2]$ with
  itself, i.e., the position consisting of two binary players sharing
  their channels; and the inclusion maps it to the top part of the
  picture.  Similarly, we have a map $[2] \into \paraof{2}$ given by
  the player $l t$ and its channels (the bottom part). The cospan
  $[2]\para [2] \to \paraof{2} \ot [2]$ is called the \emph{local fork
    move} of arity 2.

\newcommand{\longueurfigun}{.6}
\newcommand{\separation}{}

  \begin{wrapfigure}[3]{r}{0pt}
      \diagramme[stringdiag={.2}{.35}]{}{%
    \node[coordinate] (inter) at (intersection cs: %
    first line={(s)-- (s0)}, %
    second line={(s1)-- (t1)}) {} ; %
    \path[draw] (s) edge (inter) ; %
      \path[-] %
      (s1) edge (s) %
      (s0) edge (inter) %
      (s2) edge (s) %
      (s2) edge (s') %
      (t1) edge (t) %
      (t0) edge (t) %
      (t2) edge (t) %
      (t2) edge (t') %
      (t) edge (iota.west) %
      (s) edge (iota.west) %
      (s') edge (iota'.east) %
      (t') edge (iota'.east) %
      (t'0) edge (t') %
      (s'0) edge (s') %
      (s'0) edge[fore] (t'0) %
      (s0) edge[fore] (t0) %
      (s1) edge[fore] (t1) %
      (s2) edge (t2) %
      ; %
    \path[-] %
    (iota) edge[fore,densely dotted] (iota') %
    ; %
    \path (iota) --  (iota') node[coordinate,pos=.2] (iotatip) {} node[coordinate,pos=.8] (iotatip') {} ; %
      \path[-] (iota) edge[->,>=stealth,very thick] (iotatip) ; %
      \path[-] (iotatip') edge[->,>=stealth,very thick] (iota') ; %
      \foreach \x/\y in {s/t,s'/t'} \path[-] (\x) edge (\y) ; %
  }{%
    \canal{t0} \&  \& \&  \joueur{t} \& \&  \canal{t2} \& \&  %
    \joueur{t'} \& \&  \canal{t'0} \\
    \&   \canal{t1} \& \&   \& \&  \\
    \&   \& \&  \coupout{iota}{0} \& \&  \& \&  \coupin{iota'}{0} \\
    \& \& \&   \\
    \canal{s0} \&  \& \&  \joueur{s} \& \&  \canal{s2} \& \&  \joueur{s'} \& \&  \canal{s'0} \\ 
    \&  \canal{s1} %
  }{%
  }%
  \end{wrapfigure}
  For lack of space, we cannot spell out all such categories of
  elements and cospans. We give pictorial descriptions for $(m,j,n,i) = (3,3,2,1)$ of
 $\taumjni$ on the right and of $\paraln$, $\pararn$, $\iotaposmj$, $\iotanegni$, $\tickn$, 
 and $\nun$ below: \\ 
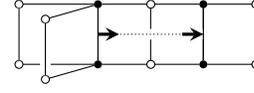
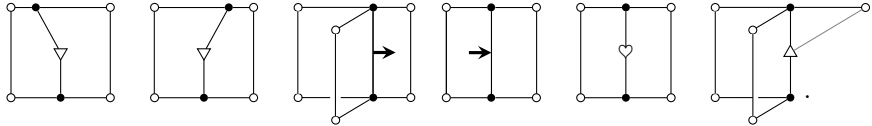
\noindent
\begin{minipage}[t][5em]{\textwidth}
    \diagramme[stringdiag={.2}{.33}]{}{ }{%
      \canal{t0} \& \joueur{t_1} \& \& \& \canal{t1}
      \\ 
      \& \&   \&  \\
      \& \ \& \\
      \coord{i0} \& \& \coupparacreux{para} \& \& \coord{i1}
      \\ 
      \& \ \& \\
      \& \&  \\
      \canal{s0} \& \& \joueur{s} \& \&
      \canal{s1} 
    }{%
      (t0) edge[-] (t_1) %
      (t1) edge[-] (t_1) %
      (s0) edge[-] (s) %
      (s1) edge[-] (s) %
      (s0) edge[-] (t0) %
      (s1) edge[-] (t1) %
      (s) edge[-] (para) %
      (para) edge[-] (t_1) %
    }
    \separation
%
    \diagramme[stringdiag={.2}{.33}]{}{ }{%
      \canal{t0} \& \& \& \joueur{t_2} \& \canal{t1}
      \\ 
      \& \&   \\
      \& \ \& \\
      \coord{i0} \& \& \coupparacreux{para}
      \\ 
      \& \ \& \\
      \& \&  \\
      \canal{s0} \& \& \joueur{s} \& \&
      \canal{s1} 
    }{%
      (t0) edge[-] (t_2) %
      (t1) edge[-] (t_2) %
      (s0) edge[-] (s) %
      (s1) edge[-] (s) %
      (s0) edge[-] (t0) %
      (s1) edge[-] (t1) %
      (s) edge[-] (para) %
      (para) edge[-] (t_2) %
    }
    \separation
%
    \diagramme[stringdiag={.3}{.5}]{baseline=($(iota.south)$)}{%
      \node[coordinate] (inter) at (intersection cs: %
      first line={(s)-- (s0)}, %
      second line={(s1)-- (t1)}) {} ; %
      \path[draw] (s) edge (inter) ; %
      \path[-] %
      (s0) edge (inter) %
      (s2) edge (s) %
      (t1) edge (t) %
      (t0) edge (t) %
      (t2) edge (t) %
      (t) edge (iota.west) %
      (s) edge (iota.west) %
      (s0) edge[fore] (t0) %
      (s1) edge[fore] (t1) %
      (s2) edge (t2) %
      (s1) edge (s) %
      ; %
      \moveout[.3]{iota}{0} %
      \foreach \x/\y in {s/t} \path[-] (\x) edge (\y) ; %
    }{%
      \canal{t0}\& \& \joueur{t} \& \canal{t2} \\
      \& \canal{t1} \\
      \& \& \coupout{iota}{0}  \\ \\
      \canal{s0}
      \& \& \joueur{s} \& \canal{s2}  \\
      \& \canal{s1} \& %
    }{%
    }%
    \separation
    \diagramme[stringdiag={.6}{\longueurfigun}]{baseline=($(in.south)$)}{
      \path[-] (a) edge (a') %
      edge (p) %
      (in.west) edge (p) edge (p') %
      (p') edge (a') edge (b') %
      (b) edge (p) edge (b') %
      ; %
      \movein[.3]{in}{180} %
      \foreach \x/\y in {p/p',a/a'} \path[-] (\x) edge (\y) ; %
    }{ %
      \canal{a'} \& \joueur{p'} \& \canal{b'} \\ %
      \coord{bin} \& \coupout{in}{0} \& \\ %
      \canal{a} \& \joueur{p} \& \canal{b} %
    }{%
    } %
    \separation
%
    \diagramme[stringdiag={.6}{\longueurfigun}]{}{ \path[-] (a) edge
      (a') %
      edge (p) %
      (tick) edge[shorten <=-1pt] (p) edge[shorten <=-1pt] (p') %
      (p') edge (a') edge (b') %
      (b) edge (p) edge (b') %
      ; %
    }{ %
      \canal{a'} \& \joueur{p'} \& \canal{b'} \\ %
      \& \couptick{tick} \& \\ %
      \canal{a} \& \joueur{p} \& \canal{b} %
    }{%
    } \separation
    \diagramme[stringdiag={.3}{.5}]{baseline=($(nu.center)$)}{%
      \node[coordinate] (inter) at (intersection cs: %
      first line={(s)-- (s0)}, %
      second line={(s1)-- (t1)}) {} ; %
      \path[draw] (s) edge (inter) ; %
      \path[-,draw] %
      (s1) edge (s) %
      (t1) edge (t) %
      (t0) edge (t) %
      (t2) edge (t) %
      (t) edge (nu) %
      (s) edge (nu) %
      (s0) edge[fore] (t0) %
      (s1) edge[fore] (t1) %
      (nu) edge[gray,very thin] (t2) %
      (inter) edge (s0) %
      ; %
      \node[diagnode,at= (s.base east)] {\ \ \ .} ; %
    }{%
      \canal{t0}
      \& \& \joueur{t} \& \& \canal{t2} \& \\
      \& \canal{t1} \&    \\
      \& \& \coupnu{nu} \& \&  \& \\
      \\
      \canal{s0}
      \& \& \joueur{s} \& \& \& \\
      \& \canal{s1} \& \& \& %
    }{%
    }%
  \end{minipage}
  In each case, the representable is the
  middle object of a cospan determined by the top and bottom parts of the
  picture. E.g., for synchronisation we have $[m]
  \mathbin{{}_{j}{\para}{}_i} [n] \xto{s} \taumjni \xot{t}[m]
  \mathbin{{}_{j}{\para}{}_i} [n]$, where $[m]
  \mathbin{{}_{j}{\para}{}_i} [n]$ denotes the position $X$ with one
  $m$-ary player $x$, one $n$-ary player $y$, such that $X(s_j)(x) =
  X(s_i)(y)$. Note that there is a crucial design choice in defining
  the legs of these cospans, which amounts to choosing initial and
  final positions for our moves. 

\begin{wrapfigure}{r}{0pt}
  \begin{minipage}[c][3em]{0.33\linewidth}
    \vspace*{-1em}
    \begin{equation}
      \label{eq:imove}
        \diag(.3,.6){%
    \&|(I)| I \\
   |(X)| X \&|(M)| M \&|(Y)| Y %
  }{%
    (I) edge (X) edge (M) edge (Y) %
    (X) edge (M) %
    (Y) edge (M) %
  }
    \end{equation}
  \end{minipage}
\end{wrapfigure}
These cospans altogether form the set of \emph{local moves}, and are
the `seeds' for (global) moves, in the following sense.  Calling an
\emph{interface} any presheaf consisting only of channels, local moves
may be equipped with a canonical interface, consisting of the channels
of their initial position. If $X \xto{s} M \xot{t} Y$ is a local move
(with final position $X$), and $I$ is its canonical interface, we
obtain a commuting diagram~\eqref{eq:imove} in $\Chatf$ (with all
arrows monic). For any morphism $I \to Z$ to some position $Z$,
pushing $I \to X$, $I \to M$, and $I \to Y$ along $I \to Z$ yields, by
universal property of pushout, a new cospan, say $X' \to M' \ot Y'$.
Letting \emph{(global) moves} be all cospans obtained in this way, and 
plays be all composites of moves in $\Cospan{\Chatf}$, we obtain, 
as promised a subbicategory $\Dccsv$.

\begin{wrapfigure}{r}{0pt}
\begin{minipage}{0.48\linewidth}
\vspace*{-1em}
    \begin{equation}
  \diagramme[stringdiag={.4}{.5}]{}{%
    \path[-,draw]
    (in) edge (s) edge (t) %
    (t0) edge (t) edge (s0) %
    (s0) edge (s) %
    (t') edge (s') edge (t0) %
    (s') edge (s0) %
    (in') edge (t') edge (u') %
    (u0) edge (u') edge (t0) %
    (u) edge (t) edge (u0) %
    ; %
    \path (in) --  (inend) node[coordinate,pos=.5] (intip) {} ; %
    \path[-] (in) edge[->,>=stealth,very thick] (intip) ; %
    \path (inend') --  (in') node[coordinate,pos=.5] (intip') {} ; %
    \path[-] (intip') edge[->,>=stealth,very thick] (in') ; %
    }{%
      \joueur{u} \& \canal{u0} \& \joueur{u'} \\
      \& \coupout{inend'}{0} \& \coupout{in'}{0} \\
      \joueur{t} \& \canal{t0} \& \joueur{t'} \\
      \coupout{in}{0} \& \coupout{inend}{0} \\
      \joueur{s} \& \canal{s0} \& \joueur{s'} %
    }{%
    }
    =
    \diagramme[stringdiag={.4}{.5}]{}{%
    \path[-,draw]
    (in) edge (s) edge (t) %
    (t0)
    edge (s0) %
    (s0) edge (s) %
    (t') edge (s') 
    (s') edge (s0) %
    (in') edge (t') edge (u') %
    (u0) edge (u') edge (t0) %
    (u) edge (t) edge (u0) %
    ; %
    \path (in) --  (inend) node[coordinate,pos=.5] (intip) {} ; %
    \path[-] (in) edge[->,>=stealth,very thick] (intip) ; %
    \path (inend') --  (in') node[coordinate,pos=.5] (intip') {} ; %
    \path[-] (intip') edge[->,>=stealth,very thick] (in') ; %
    }{%
      \joueur{u} \& \canal{u0} \& \joueur{u'} \\
      \& \coupout{inend'}{0} \& \coupout{in'}{0} \\
      \coord{t} \& \coord{t0} \& \coord{t'} \\
      \coupout{in}{0} \& \coupout{inend}{0} \\
      \joueur{s} \& \canal{s0} \& \joueur{s'} %
    }{%
    }    
    =
    \diagramme[stringdiag={.4}{.5}]{}{%
    \path[-,draw]
    (in) edge (s) edge (u) %
    (in') edge (s') edge (u') %
    (u0) edge (u) %
    edge (u') %
    edge (s0) %
    (s0) edge (s) %
    edge (s') %
    ; %
    \path (in) --  (end) node[coordinate,pos=.5] (intip) {} ; %
    \path[-] (in) edge[->,>=stealth,very thick] (intip) ; %
    \path (end) --  (in') node[coordinate,pos=.5] (intip') {} ; %
    \path[-] (intip') edge[->,>=stealth,very thick] (in') ; %
    }{%
      \joueur{u} \& \canal{u0} \& \joueur{u'} \\
     \coupout{in}{0}  \& \coupout{end}{0} \& \coupout{in'}{0} \\
      \joueur{s} \& \canal{s0} \& \joueur{s'} %
    }{%
    }    \label{eq:play}
\end{equation}
  \end{minipage}
\end{wrapfigure}
Passing from local to global moves allows moves to occur in larger
positions. Furthermore, we observe that plays feature some
concurrency. For instance, composing two global moves as on the right,
we obtain a play in which the order of appearance of moves is no
longer visible.  In passing, this play embeds into a synchronisation,
but is not one, since the input and output moves are not related. This
play may be understood as each player communicating with the outside
world.
We conclude with a useful classification of moves.
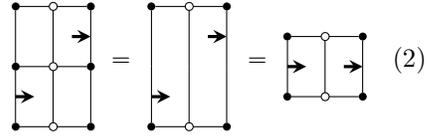
\begin{definition}
  A move is \emph{full} iff it is neither a left nor a right fork. We
  call $\F$ the graph of global, full moves.
\end{definition}
Intuitively, a move is full when its final position contains all
possible avatars of involved players.

\section{Behaviours, strategies, and fair testing}\label{sec:strats}
\subsection{Behaviours}

\begin{wrapfigure}{r}{0pt}
      \diag|baseline= (m-1-1.base) |{%
      U \& U' \\
      X \& X' %
    }{%
      (m-1-1) edge[labelu={}] (m-1-2) %
      (m-2-1) edge[labell={}] (m-1-1) %
      (m-2-1) edge[labeld={}] (m-2-2) %
      (m-2-2) edge[labelr={}] (m-1-2) %
    }
\end{wrapfigure}
Recall from HP the category $\E$ whose objects are maps $U \ot X$ in
$\Chat$, such that there exists a play $Y \to U \ot X$, i.e., objects
are plays, where we forget the final position. Its morphisms $(U
\ot X) \to (U' \ot X')$ are commuting diagrams as on the right with
all arrows monic.  Morphisms $U \to U'$ in $\E$ represent extensions
of $U$, both spatially (i.e., embedding into a larger position) and
dynamically (i.e., adding more moves).

We may relativise this category $\E$ to a particular position $X$,
yielding a category $\E(X)$ of plays on $X$: take the fibre over $X$
of the functor $\cod \colon \E \to \Dccsh$ mapping any play $U \ot X$
to its initial position $X$. The objects of $\E(X)$ are just plays $(U
\ot X)$ on $X$, and morphisms are morphisms of plays whose lower
border is $\id_X$.  This leads to a category of `naive' strategies,
called behaviours.
\begin{definition}
  The category $\Beh{X}$ of \emph{behaviours} on $X$ is the category
  $\Pshf{\E(X)}$ of presheaves of finite sets on $\E (X)$.
\end{definition}
Behaviours suffer from the deficiency of allowing unwanted cooperation
between players. HP (Example 12) exhibits a behaviour where players
choose with whom they synchronise, which clearly is not allowed in CCS.

\subsection{Strategies}
To rectify this, we consider the full subcategory $\EVi$ of $\E$
consisting of \emph{views}, i.e., compositions of basic local
moves. We relativise views to a position $X$, as follows.  Let, for
any $n \in \Nat$, $[n]$ denote the single $n$-ary player, i.e., a
single player connected to $n$ distinct channels.  Players of $X$ are
in 1-1 correspondence with pairs $(n,x)$, with $x \colon [n] \to X$ in
$\Dccsh$.  Relativisation of $\EVi$ to $X$ is given by the category
$\EVi_X$ with as objects all pairs $(V,x)$, where $x \colon [n] \to
X$, and $V$ is a view with initial position $[n]$.  Morphisms are
induced by those of $\E$.

    \begin{definition}
     The category $\SS_X$ of \emph{strategies} on $X$ is the category
     $\OPsh{\EVi_X}$ of presheaves of finite ordinals on $\EVi_X$.
    \end{definition}
    
    \begin{wrapfigure}[4]{r}{0pt}
      \diagramme{baseline=($(m-1-1.base)$)}{}{%
        \op{\EVi_X} \& \op{\E_X} \& \op{\E(X)} \\
        \ford \& \set, %
      }{%
        (m-1-1) edge[labell={S}] (m-2-1) %
        edge[into] (m-1-2) %
        (m-2-1) edge[into,labelu={i}] (m-2-2) %
        (m-1-2) edge[labell={S'}] (m-2-2) %
        (m-1-3) edge[linto,labeld={j}] (m-1-2) %
        edge[labelbr={\exta{X}{S}}] (m-2-2)
      }
    \end{wrapfigure}
    This rules out undesired behaviours. Recall from HP how to map
    strategies to behaviours: let first $\E_X$ be the category
    obtained as $\EVi_X$ from all plays instead of just views.  Then,
    starting from a strategy $S$, let $S'$ be obtained by right Kan
    extension of $i \rond S$ (by $\op{\EVi_X} \into \op{\E_X}$ being
    full and faithful), and let $\exta{X}{S} = S' \rond j$.  The assignment $S
    \mapsto \exta{X}{S}$ extends to a full and faithful functor $\extafun{X}
    \colon \SS_X \to \Beh{X}$. Furthermore, $\extafun{X}$ admits a
    left adjoint, which we call \emph{innocentisation}, maping naive
    strategies (behaviours) to innocent ones. By standard
    results~\cite{MacLane:cwm}, we have for any $S$: $\exta{X}{S} (U)
    = \int_{v \in \EVi_X} S (v)^{\E_X (v,U)}.$ Equivalently,
    $\exta{X}{S} (U)$ is a limit of $\op{(\EVi_X / U)} \xto{\dom}
    \op{\EVi_X} \xto{S} \ford \into \set.$

\subsection{Decomposition: a syntax for strategies}\label{subsec:syntax}
Our definition of strategies is rather semantic in flavour. Indeed,
presheaves are akin to domain theory. However, they also lend
themselves well to a syntactic description.  First, it is shown in HP
that strategies on an arbitrary position $X$ are in 1-1 correspondence
with families of strategies indexed by the players of $X$. Recall that
$[n]$ is the position consisting of one $n$-ary player, and that
players of $X$ may be defined as elements of $\Pl (X) = \sum_{n \in
  \Nat} \Dccsh ([n],X)$.
\begin{proposition}
We have  $\SSX \iso \prod_{(n,x) \in \Pl (X)} \SSn$. 
For any $S \in \SSX$, we denote by $S_{(n,x)}$ the component corresponding
to $(n,x) \in \Pl (X)$ under this isomorphism.
\end{proposition}
This result yields a construction letting two strategies interact
along an \emph{interface}, i.e., a position consisting only of
channels. This will be the basis of our semantic definition of fair
testing equivalence. Consider any pushout $Z$ of $X \ot I \to Y$ where
$I$ is an interface. We have
\begin{corollary}
  $\SS_Z \iso \SSX \times \SSY$.
\end{corollary}
\begin{proof}
  We have $\EVi_Z \iso \EVi_X + \EVi_Y$, and conclude by universal
  property of coproduct.
\end{proof}
We denote by $[S,T]$ the image of $(S,T) \in \SSX \times \SSY$ under
this isomorphism.

So, strategies over arbitrary positions may be decomposed into
strategies over `typical' players $[n]$. Let us now explain that
strategies over such players may be further decomposed. For any
strategy $S$ on $[n]$ and basic move $b \colon [n'] \to [n]$, let the
\emph{residual} $S \cdot b$ of $S$ after $b$ be the strategy playing
like $S$ after $b$, i.e., for all $v \in \EVi_{[n']}$, $(S \cdot b)(v)
= S (b \vrond v)$, where $\vrond$ denotes composition in $\Dccsv$. $S$
is almost determined by its residuals. The only information missing
from the $S\cdot b$'s to reconstruct $S$ is the set of initial states
and how they relate to the initial states of each $(S \cdot b)$.
Thus, for any position $X$, let $\id^v_{X}$ denote the identity play
on $X$ (i.e., nothing happens).  For any initial state $\state \in
S(\id_{[n]})$, let $\restr{S}{\state}$ be the restriction of $S$ to
states derived from $\state$, i.e., for all $v$, those $\state' \in
S(v)$ which are mapped to $\state$ under the restriction $S(!) \colon
S(v) \to S(\id_{[n]})$.  $S$ is determined by its set $S(\id_{[n]})$
of initial states, plus the function $(\state, b) \mapsto
(\restr{S}{\state} \cdot b)$. Since $S (\id_{[n]})$ is a finite
ordinal $m$, we have for all $n$:
\begin{theorem}\label{thm:stratcoalg}
   $\SS_{[n]} \iso \sum_{m \in \Nat} (\prod_{b \colon [n'] \to [n]} \SS_{[n']})^m
   \iso (\prod_{b \colon [n'] \to [n]} \SS_{[n']})^\star$.
\end{theorem}
This result may be understood as saying that strategies form a
fixpoint of a certain (polynomial~\cite{Kock01012011}) endofunctor of
$\Set / \DI$, where $\DI$ is the set of `typical' players $[n]$.  This
may be strengthened to show that they form a terminal coalgebra, i.e,
that they are in bijection with infinite terms in the following typed
grammar, with judgements $n \vdashdefinite D$ and $n \vdash S$, where $D$ is
called a \emph{definite prestrategy} and $S$ is a \emph{strategy}:%
\begin{mathpar}
\inferrule{\ldots \ n_b \vdash S_b \ \ldots \ {(\forall b \colon [n_b] \to [n] \in \MMMB_n)}
}{
n \vdashdefinite \langle (S_b)_{b \in \MMMB_n} \rangle
}
\and
\inferrule{\ldots \  n \vdashdefinite D_i \  \ldots \ (\forall i \in m)}{
n \vdash \oplus_{i \in m} D_i}~(m \in \Nat),
\end{mathpar}%
where $\MMMB_n$ denotes the set of all isomorphism classes of basic
moves from $[n]$. We need to use isomorphism classes here, because
strategies may not distinguish between different, yet isomorphic basic
moves.  This achieves the promised syntactic description of
strategies. We may readily define the translation of CCS processes,
coinductively, as follows. For processes with channels in $\Gam$, we
define
\begin{mathpar}
\begin{array}[t]{rcl}
  \Transl{\sum_{i \in n} \alpha_i.P_i} & = & \langle b \mapsto
      \oplus_{\ens{i \in n \aalt b = \Transl{\alpha_i}}} \Transl{P_i}
      \rangle \\
      \Transl{\nu a.P} & = & 
      \langle
            \nugam  \mapsto  \Transl{P},
            \_  \mapsto  \emptyset 
            \rangle \\
      \Transl{P\para Q} & = & 
 \langle
            \paralgam  \mapsto  \Transl{P}  ,
            \parargam  \mapsto  \Transl{Q}  ,
            \_  \mapsto  \emptyset 

 \rangle 
\end{array}
\and \begin{array}[t]{rcl}
    \Transl{a} & = & \iotaneg{\Gam,a} \\
    \Transl{\abar} & = & \iotapos{\Gam,a} \\
    \Transl{\tick} & = & \tickgam.
\end{array}
\end{mathpar}
E.g., $a.P + a.Q + \bar{b}.R$ is mapped to
$\langle \iotaneg{\Gam,a} \mapsto (\Transl{P} \oplus \Transl{Q}),
\iotapos{\Gam,b} \mapsto \Transl{R}, \_ \mapsto \emptyset \rangle.$

\subsection{Semantic fair testing}
We may now recall our semantic analogue of fair testing equivalence.
\begin{definition}
 \emph{Closed-world} moves are (the global variants of)
  $\new$,$\tick$,$\paran$, and $\taunimj$. A play is \emph{closed-world}
  when it is a composite of closed-world moves.
\end{definition}

Let a closed-world play be \emph{successful} when it contains a
$\tick$ move.  Let then $\bbot_Z$ denote the set of behaviours $B$
such that for any unsuccessful, closed-world play $U \ot Z$ and
$\state \in B (U)$, there exists $f \colon U \to U'$, with $U'$
closed-world and successful, and $\state' \in B (U')$ such that
$B(f)(\state') = \state$.  Finally, let us say that a triple
$(I,h,S)$, for any $h \colon I \to X$ and strategy $S \in \SSX$,
\emph{passes} the test consisting of a morphism $k \colon I \to Y$ of
positions and a strategy $T \in \SSY$ iff $\exta{Z}{[S,T]} \in
\bbot{}_Z$, where $Z$ is the pushout of $h$ and $k$.  Let
$S^{\bbot{}}$ denote the set of all such $(k,T)$.
\begin{definition}
  For any $h \colon I \to X$, $h' \colon I \to X'$, $S \in \SSX$, and
  $S' \in \SS_{X'}$, $(I,h,S) \faireq (I,h',S')$ iff
  $(I,h,S)^{\bbot{}} = (I,h',S')^{\bbot{}}$.
\end{definition}
This yields an equivalence relation, analogous to standard fair
testing equivalence, which we hence also call fair testing
equivalence.

We have defined a translation $\Translfun$ of CCS processes to
strategies, which raises the question of whether it preserves or
reflects fair testing equivalence. The rest of the paper is devoted to
proving that it does both.


\section{Playgrounds and main result}\label{sec:main}

\subsection{Playgrounds: a theory of individuality and atomicity}\label{subsec:playgrounds}
\begin{wrapfigure}{r}{0pt}
  \Diag{%
    \twocellbr{m-2-1}{m-1-1}{m-1-2}{\alpha} %
    \twocellbr{m-2-2}{m-1-2}{m-1-3}{\alpha'} %
    \twocellbr{m-3-1}{m-2-1}{m-2-2}{\beta} %
    \twocellbr{m-3-2}{m-2-2}{m-2-3}{\beta'} %
  }{%
    X \& X' \& X'' \\
    Y \& Y' \& Y'' \\
    Z \& Z' \& Z'',
  }{%
    (m-1-1) edge[labelu={h}] (m-1-2) %
    edge[labell={u}] (m-2-1) %
    (m-2-1) edge[labelu={h'}] (m-2-2) %
    (m-1-2) edge[labell={u'}] (m-2-2) %
    (m-1-2) edge[labelu={k}] (m-1-3) %
    (m-2-2) edge[labelu={k'}] (m-2-3) %
    (m-1-3) edge[labelr={u''}] (m-2-3) %
    (m-2-1) 
    edge[labell={v}] (m-3-1) %
    (m-3-1) edge[labeld={h''}] (m-3-2) %
    (m-2-2) edge[labell={v'}] (m-3-2) %
    (m-3-2) edge[labeld={k''}] (m-3-3) %
    (m-2-3) edge[labelr={v''}] (m-3-3) %
  }
\end{wrapfigure}
We start by trying to give an idea of the notion of playground.
To start with, we organise the game into a \emph{(pseudo) double
  category}~\cite{GrandisPare,GarnerPhD}.  This is a weakening of
Ehresmann's double categories~\cite{Ehresmann:double2}, where one
direction has non strictly associative composition. Although we
consider proper pseudo double categories, we often may treat them
safely as double categories.  A pseudo double category $\D$ consists
of a set $\ob (\D)$ of \emph{objects}, shared by two categories $\Dh$
and $\Dv$. $\Dh$ is called the \emph{horizontal} category of $\D$, and
$\Dv$ is the \emph{vertical} category. Composition in $\Dh$ is denoted
by $\rond$, while we use $\vrond$ for $\Dv$. $\D$ is furthermore
equipped with a set of \emph{double cells} $\alpha$, which have
vertical, resp.\ horizontal, domain and codomain, denoted by $\domv
(\alpha)$, $\codv (\alpha)$, $\domh (\alpha)$, and $\codh
(\alpha)$. We picture this as, e.g., $\alpha$ above, where $u = \domh
(\alpha)$, $u' = \codh (\alpha)$, $h = \domv (\alpha)$, and $h' =
\codv (\alpha)$. $\D$ is furthermore equipped with operations for
composing double cells: $\rond$ composes them along a common vertical
morphism, $\vrond$ composes along horizontal morphisms. Both vertical
compositions (of morphisms and double cells) may only be associative
up to coherent isomorphism. The full axiomatisation is given by
Garner~\cite{GarnerPhD}, and we here only mention the
\emph{interchange law}, which says that the two ways of parsing the
above diagram coincide: $(\beta' \rond \beta) \vrond (\alpha' \rond
\alpha) = (\beta' \vrond \alpha') \rond (\beta \vrond \alpha)$.

\begin{example}
Returning to the game, we have seen that positions are the objects of
the category $\Dccsh$, whose morphisms are embeddings of positions.
But positions are also the objects of the bicategory $\Dccsv$, whose
morphisms are plays. \linebreak\noindent
\begin{minipage}[t]{\linewidth}
\vspace*{-.5em}\begin{wrapfigure}[6]{r}{0pt}
  \begin{minipage}[t]{0.2\linewidth}
\vspace*{-.55em} \hspace*{-1em}   \diag|baseline=($(m-1-1.south)$)|{%
      X \& X' \\
      U \& V \\
      Y \& Y' %
    }{%
      (m-1-1) edge[labeld={h}] (m-1-2) %
      (m-2-1) edge[labelo={k}] (m-2-2) %
      (m-3-1) edge[labelu={l}] (m-3-2) %
      (m-1-1) edge[labell={s}] (m-2-1) %
      (m-1-2) edge[labelr={s'}] (m-2-2) %
      (m-3-1) edge[labell={t}] (m-2-1) 
      (m-3-2) edge[labelr={t'}] (m-2-2) %
    }
  \end{minipage}
\end{wrapfigure}
  It should seem natural to define a pseudo double category structure
  with double cells given by commuting diagrams as on the right in
  $\Chat$. Here, $Y$ is the initial position and $X$ is the final one;
  all arrows are mono.  This indeed forms a pseudo double category
  $\Dccs$. Furthermore, for any double category $\D$, let $\DH$ be the
  category with objects all morphisms of $\Dv$, and with morphisms $u
  \to u'$ all double cells $\alpha$ such that $\dom_h(\alpha) = u$ and
  $\cod_h(\alpha) = u'$.  A crucial feature of $\Dccs$ is that the
  canonical functor $\codv \colon \DH \to \Dh$ mapping any such
  $\alpha$ to $\codv(\alpha)$ is a Grothendieck
  fibration~\cite{Jacobs}.  This means that one may canonically
  `restrict' a play, say $u' \colon X' \to Y'$, along a horizontal
  morphism $h' \colon Y \to Y'$, and obtain a universal cell as
  $\alpha$ above, in a suitable sense.
\end{minipage}
\end{example}

\begin{wrapfigure}[4]{r}{0pt}
  \diagramme{}{}{%
    |(d)| d \& |(Y)| Y \\
    |(dMy)| d^{y,M} \& |(X)| X %
  }{%
    (d) edge[labelu={y}] (Y) %
    edge[dashed,twol={v^{y,M}}] (dMy) %
    (Y) edge[twor={M}] (X) %
    (dMy) edge[dashed,labelu={y^M}] (X) %
    (l) edge[cell=.3,dashed,labelu={\scriptstyle \alpha^{y,M}}] (r) %
  }
\end{wrapfigure}
Playgrounds are pseudo double categories with extra data and axioms,
the first of which is that $\codv$ should be a fibration. To give a
brief idea of further axioms, a playground $\D$ is equipped with a set
of objects $\DI$, called \emph{individuals}, which correspond to our
`typical' players above. Let $\Pl(X) = \sum_{d \in \DI} \Dh (d,X)$
denote the set of players of $X$. It also comes with classes $\F$ and
$\B$ of \emph{full}, resp.\ \emph{basic} moves; and every play (i.e.,
vertical morphism) is assumed to admit a decomposition into moves in
$\F \cup \B$ (hence \emph{atomicity}). Basic moves are assumed to have
individuals as both domain and codomain, and \emph{views} are defined
to be composites of basic moves. The crucial axiom for innocence to
behave well assumes that, for any position $Y$ and player $y \colon d
\to Y$, there exists a cell $\alpha^{y,M}$ as above, with $v^{y,M}$ a
view, which is unique up to canonical isomorphism of such.
Intuitively: any player in the final position of a play has an
essentially unique view of the play. A last, sample axiom shows how
some sequentiality is enforced, which is useful to tame the
concurrency observed in~\eqref{eq:play}. It says that any double cell
as in the center below, where $b$ is a basic move and $M$ is any move,
decomposes in exactly one of the forms on the left and right:
  \begin{center}
      \begin{minipage}[t]{0.21\textwidth}
        \Diag{%
          \twocell[.4][.3]{B}{A}{X}{}{celllr={0.0}{0.0},bend
            right=30,labelbr={\alpha_1}} %
          \twocell[.4][.3]{C}{B}{Y}{}{celllr={0.0}{0.0},bend
            right=20,labelbr={\alpha_2}} %
        }{%
          |(A)| A \& |(X)| X \\
          |(B)| B \& |(Y)| Y \\
          |(C)| C \& |(Z)| Z %
        }{%
          (A) edge (B) %
          edge (X) %
          (B) edge (Y) %
          edge (C) %
          (C) edge (Z) %
          (X) edge (Y) %
          (Y) edge 
          (Z) %
        }
      \end{minipage}
      \hfil $\leftsquigarrow$ \hfil
      \begin{minipage}[t]{0.21\textwidth}
        \Diag{%
          \twocellbr{B}{A}{X}{\alpha} %
        }{%
          |(A)| A \& |(X)| X \\
          |(B)| B \& |(Y)| Y \\
          |(C)| C \& |(Z)| Z %
        }{%
          (A) edge[labelu={h}] (X) %
          edge[labell={w}] (B) %
          (B) edge[labell={b}] (C) %
          (X) edge[labelr={u}] (Y) %
          (Y) edge[labelr={M}] (Z) %
          (C) edge[labeld={k}] (Z) %
        }
      \end{minipage}
    \hfil $\rightsquigarrow$ \hfil
      \begin{minipage}[t]{0.21\textwidth}
        \Diag {%
          \twocell[.4][.3]{B}{A}{X}{}{celllr={0.0}{0.0},bend
            right=30,labelbr={\alpha_1}} %
          \twocell[.3][.4]{C}{Y}{Z}{}{celllr={0.0}{0.0},bend
            right=20,labeld={\alpha_2}} %
        }{%
          |(A)| A \& |(X)| X \\
          |(B)| B \& |(Y)| Y \\
          |(C)| C \& |(Z)| Z. %
        }{%
          (A) edge (B) %
          edge (X) %
          (B) 
          edge (C) %
          (C) edge (Z) %
          edge (Y) %
          (X) edge (Y) %
          (Y) 
          edge (Z) %
        }
      \end{minipage}
  \end{center}
  The idea is that, $C$ being an individual, if $M$ has a non-trivial
  restriction to $C$, then $b$ must be one of its views.
  Again, for the formal definition, see~\cite{HirschoDoubleCats}.
  \begin{proposition}
    $\Dccs$ forms a playground (basic moves being the \emph{local} ones).
  \end{proposition}

\subsection{Syntaxes and labelled transition systems}\label{subsec:ltss}
Notions of residuals and restrictions defined above for CCS are easily
generalised to arbitrary playgrounds. They lead to the exact same
syntax as in the concrete case (below
Theorem~\ref{thm:stratcoalg}). They further yield a first, naive
\lts{} over full moves for strategies. The intuition is that there is
a transition $S \xto{M} S'$, for any full move $M$, when $S \cdot M =
S'$. (Residuals $S \cdot M$ are here defined analogously to the case
of basic moves $S \cdot b$ above.) An issue with this \lts{} is that
$S \cdot M$ may have several possible initial states, and we have seen
that it makes more sense to restrict to a single state before taking
residuals.  We thus define our \lts{} $\SSS_\D$ to have as vertices
pairs $(X,S)$ of a position $X$ and a \emph{definite} strategy $S$,
i.e., a strategy with exactly one initial state (formally, $S_{(d,x)}
(\id_d) = 1$ for all $(d,x) \in \Pl (X)$ --- recalling that $\id_d$ is
an (initial) object in $\EVi_d$).  We then say that there is a
transition $(X,S) \xto{M} (X',S')$ for any full move $M \colon X' \to
X$, when $S' = \restr{(S \cdot M)}{\state'}$, for some initial state
$\state'$ of $S \cdot M$.

\begin{example}
  %
  Consider a strategy of the shape $S = \with{\pararn \mapsto
    S_1, \paraln \mapsto S_2, \_ \mapsto \emptyset}$ on $[n]$, with
  definite $S_1$ and $S_2$.  There is a $\paran$ transition to the
  position with two $n$-ary players $x_1$ and $x_2$, equipped with
  $S_1$ and $S_2$, respectively. If now $S_1$ and $S_2$ are not
  definite, any $\paran$ transition has to pick initial states
  $\state_1 \in S_1(\id_{[n]})$ and $\state_2 \in S_2(\id_{[n]})$,
  i.e., $S \xto{\paran}
  [\restr{(S_1)}{\state_1}]\para[\restr{(S_2)}{\state_2}]$.  Here, we
  use a shorthand notation for pairs $(X,S)$, defined as follows.
  First, for any strategy $S$ over $[n]$ and position $X$ with exactly
  one $n$-ary player $x$ and names in $\Gam$, we denote by $\Gam
  \vdash [x:S](a_1,\ldots, a_n)$ the pair $(X,S)$, where $a_i =
  X(s_i)(x)$, for all $i \in n$. If now $X$ has several players, say
  $x_1, \ldots, x_p$, of respective arities $n_1, \ldots, n_p$, and
  $S_1, \ldots, S_p$ are strategies of such arities, we denote by
  $\Gam \vdash [x_1:S_1](a^1_1, \ldots, a^1_{n_1}) \para \ldots \para
  [x_p:S_p](a^p_1, \ldots, a^p_{n_p})$ the pair $(X,
  [S_1,\ldots,S_p])$.  When they are irrelevant, we often omit $\Gam$,
  the $x_j$'s, and the $a^j_i$'s, as in our example.  
\end{example}

\begin{wrapfigure}{r}{0pt}
      \begin{minipage}{.2\textwidth}
        \centering
  \inferrule{ \ldots \ d_x \vdash T_x \ \ldots}{ d \vdash M \with{(T_x)_{x
        \in \Pl(M)}}} 

        \inferrule{%
          \ldots \ d_i \vdash T_i \ \ldots \ (\forall i \in n) %
        }{%
          d \vdash {\textstyle \sum}_{i \in n} M_i.T_i %
        }
    \end{minipage}
\end{wrapfigure}
Beyond the one for strategies, there is another syntax one can derive
from any playground. Instead of relying on basic moves as before, one
now relies on full moves. Thinking of full moves as inference rules
(e.g., in natural deduction), the premises of the rule for any full $M
\colon X \to Y$ should be those players $(d_x,x)$ of $X$ whose view
through $M$ is non-trivial, i.e., is a basic move. We call this set of
players $\Pl(M)$. The natural syntax rule is thus the first one above
(glossing over some details), which defines \emph{process terms} $T$.
We add a further rule for guarded sum allowing to choose between
several moves. One has to be a little careful here, and only allow
moves $M \colon X \to Y$ such that $\Pl(M)$ is a singleton. This
yields the second rule above, where $n \in \Nat$, and $\forall i \in
n$, $M_i$ is such a move and $d_i$ is the arity of the unique element
of $\Pl(M_i)$. Calling $\TT_d$ the set of infinite terms for this
syntax, there is a natural translation map $\translfun \colon \TT_d
\to \SS_d$ to strategies, for all $d \in \DI$, which looks a lot like
$\Translfun$, and \anlts{} $\TTT_\D$, whose vertices are pairs $(X,T)$
of a position $X$, with $T \in \prod_{d,x \in \Pl (X)} \TT_d$. The
main result on playgrounds is
\begin{theorem}\label{thm:bisim}
  The map $\translfun \colon \TTT_\D \to \SSS_\D$ is a functional, strong bisimulation.
\end{theorem}

\subsection{Change of base and main result}\label{subsec:cob}
The \lts{} $\SSS_{\Dccs}$ obtained for $\Dccs$ is much too fine to be
relevant for bisimilarity to make behavioural sense. E.g., the
translations of $a \para b$ and $b \para a$ are not bisimilar. Indeed,
labels, i.e., full moves in $\F_{\Dccs}$, bear the information of
which player is involved in the transition. So both strategies have a
$\paragam$ translation to a position with two $\Gam$-ary players, say
$x_1$ and $x_2$. But then, $a \para b$ has a transition where $x_1$
plays an input on $a$, which $b \para a$ cannot match. Refining the
above notation, and omitting $\Translfun$, we may write the former
transitions as $[a \para b] \xto{\paragam} [a]\para [b]
\xto{x_1,\iotaneg{\Gam,a}} [0]\para [b].$ There is another problem
with this \lts{}, namely that there are undue transitions. E.g., we
have $[\nu a.a] \xto{\nu_0} [a] \xto{\iotaneg{(a),a}} 0$.  The
transition system does not yet take privacy of channels into account.

Let us first rectify the latter deficiency. To this end, we pull back
our \lts{} $\SSS_{\Dccs} \to \F_{\Dccs}$ along a morphism of graphs $\LLL \to
\F_{\Dccs}$ defined as follows. Let $\LLL$ have \emph{interfaced positions} as
vertices, i.e., morphisms $h \colon I \to X$ from an interface to a
position. $I$ specifies the public channels, and hence we let edges $h
\to h'$ be commuting diagrams of the shape~\eqref{eq:imove}, where $M$
may be any full move ($X$ being the final position), except inputs and
outputs on a channel outside the image of $I$.  We then
straightforwardly define $\chi \colon \LLL \to \F_{\Dccs}$ to map $h$ to $X$
and any diagram above to $M$. The pullback $\SSSL_{\Dccs} \to \LLL$ of
$\SSS_{\Dccs}$ along $\chi$ is rid of undue communications on private
channels.

To rectify the other deficiency mentioned above, recalling from
Definition~\ref{defn:A} that $\A$ is the alphabet for CCS, we define a
morphism $\xi \colon \LLL \to \A$ by mapping $(I \to X)$ to its set $I
(\star)$ of channels, and any $M$ to
(1) $\tick$ if $M$ is a tick move,
(2) $\id$ if $M$ is a synchronisation, a fork, or a channel creation,
(3) $a$ if $M$ is an input on $a \in I (\star)$,
(4) $\abar$ if $M$ is an output on $a \in I (\star)$.
(Positions are formally defined as presheaves to $\set$, hence
channels directly form a finite ordinal number.)  It is here crucial
to have restricted attention to $\LLL$ beforehand, otherwise we would
not know what to do with communications on private channels. Let
$\rc{\SSS_{\Dccs}} = \pbang{\xi}{\SSSL_{\Dccs}}$ be the post-composition
of $\SSSL_{\Dccs} \to \LLL$ with $\xi$.

The obtained \lts{} $\SSSAccs \to \A$ is now ready for our purposes.
Proceeding similarly for the \lts{} $\TTT_{\Dccs}$ of process terms,
we obtain a strong, functional bisimulation $\translfun \colon
\ob(\TTTAccs) \to \ob (\SSSAccs)$ over $\A$. We then prove that
$\theta \colon \ob(\ccs) \into \ob (\TTTAccs)$ is included in weak
bisimilarity over $\A$, and, easily, that $\Translfun = \translfun
\rond \theta$. 
\begin{corollary}\label{cor:wbisim}
  For all $P$, $P \wbisim_\A \Transl{P}$.
\end{corollary}
Furthermore, we prove that $\faireq$ coincides with the standard,
\lts{}-based definition of fair testing, i.e., $P \faireqs Q$ iff for
all sensible $T$, $(P\para T \in \bot_s) \Leftrightarrow (Q\para T
\in \bot_s)$, where $P \in \bot_s$ iff any $\tick$-free reduction
sequence $P \Rightarrow P'$ extends to one with $\tick$. 
To obtain our main
result, we finally generalise an observation of Rensink and
Vogler~\cite{DBLP:journals/iandc/RensinkV07}, which essentially says
that for fair testing equivalence in CCS, it is sufficient to consider
a certain class of tree-like tests, called \emph{failures}. We first slightly
generalise the abstract setting of De Nicola and
Hennessy~\cite{DBLP:journals/tcs/NicolaH84} for testing equivalences,
e.g., to accomodate the fact that strategies are indexed over
interfaces.  This yields a notion of \emph{effective graph}.  We then
show that, for any effective graph $G$ over an alphabet $A$, the
result on failures goes through, provided $G$ \emph{has enough
  $A$-trees}, in the sense that, up to mild conditions, for any tree
$t$ over $A$, there exists $x \in G$ such that $x \wbisima t$.
Consequently, for any relation $R \colon G \modto G'$ between two such
effective graphs with enough $A$-trees, if $R$ is included in weak
bisimilarity over $A$, then $R$ preserves and reflects fair testing
equivalence. We thus obtain our main result:


\begin{theorem}\label{thm:final}
  For any $\Gam \in \Nat$, let $I_\Gam$ be the interface consisting of
  $\Gam$ channels, and $h_\Gam \colon I_\Gam \to [\Gam]$ be the
  canonical inclusion.  For any CCS processes $P$ and $Q$ over $\Gam$,
  we have $P \faireqs Q$ iff $(I_\Gam,h_\Gam,\Transl{P}) \faireq
  (I_\Gam,h_\Gam,\Transl{Q})$.
\end{theorem}

\begin{remark}
  Until now, we have considered arbitrary, infinite CCS processes. Let
  us now restrict ourselves to recursive processes (e.g., in the sense
  of HP). We obviously still have that $\Transl{P} \faireq \Transl{Q}$
  implies $P \faireqs Q$.  The converse is less obvious and may be
  stated in very simple terms: suppose you have two recursive CCS
  processes $P$ and $Q$ and a test process $T$, possibly
  non-recursive, distinguishing $P$ from $Q$; is there any
  recursive $T'$ also distinguishing $P$ from $Q$?
\end{remark}

\bibsettings{\tinybibsettings}
\bibliography{../common/bib}
\bibliographystyle{jbwc}

\end{document}